\documentclass[conference]{IEEEtran}
\usepackage{balance}
\usepackage{cite}
\usepackage{amsmath,amssymb,amsfonts}
\usepackage{ntheorem}

\newtheorem{theorem}{Theorem}
\newtheorem{definition}{Definition}
\newtheorem{lemma}{Lemma}
\newtheorem{corollary}{Corollary}

\newtheorem{conjecture}{Conjecture}
\newenvironment{proof}{\noindent\textbf{Proof. }}{\hfill$\Box$}
\usepackage{graphicx}
\usepackage{xcolor}
\usepackage{booktabs}
\usepackage{array}
\usepackage{url}
\usepackage{xspace}
\usepackage{tikz}
\usetikzlibrary{arrows.meta,positioning,fit,calc,shapes.geometric}
\usepackage[hidelinks]{hyperref}
\usepackage{marvosym}
\usepackage[hidelinks]{hyperref}
\hypersetup{pdfauthor={Sirui Liu, Li Que, Zongpeng Li, Baochun Li},pdftitle={On the Multiple-Unicast Conjecture: \\Beyond Cut Metrics}}

\newcommand{\MET}{\operatorname{MET}}
\newcommand{\NC}{\mathsf{NC}}
\newcommand{\MCF}{\mathsf{MCF}}
\newcommand{\holes}{\mathcal{F}}

\newcommand{\etal}{{\em et al.}\xspace}
\definecolor{sessionblue}{RGB}{35,90,170}
\definecolor{sessionorange}{RGB}{222,112,28}
\definecolor{holecolor}{RGB}{190,35,45}
\definecolor{softblue}{RGB}{226,238,252}
\definecolor{softorange}{RGB}{253,237,214}
\definecolor{softgreen}{RGB}{226,243,229}

\begin{document}

\title{On the Multiple-Unicast Conjecture:\\ Beyond Cut Metrics}

\author{
Sirui Liu$^{\dagger,\ddagger}$,
Li Que$^{\dagger}$,
\href{mailto:zongpeng@tsinghua.edu.cn}{Zongpeng Li}$^{\ddagger,\dagger \,}$%
\href{mailto:zongpeng@tsinghua.edu.cn}{\textsuperscript{\Letter}},
Baochun Li$^{\S}$
\\
$^{\dagger}$\emph{Tsinghua University}
\qquad
$^{\ddagger}$\emph{Quan Cheng Laboratory}
\qquad
$^{\S}$\emph{University of Toronto}
}

\maketitle

\begin{abstract}
Network coding allows intermediate nodes to encode received messages before
transmission. The multiple-unicast conjecture asserts that coding has no
throughput advantage over fractional routing for independent unicast sessions
in any undirected network. Despite more than two decades of sustained study,
this central open problem remains unresolved. The conjecture is deeply
connected to computational complexity: a proof would yield long-sought lower
bounds for fundamental problems. To study the conjecture, this paper develops
a unified metric framework from the perspective that the basic objects behind
the comparison between coding and routing are not cuts alone, but graph
metrics. Using this framework, we prove the conjecture for three new classes
of undirected networks: (a)~networks with at most five terminal locations;
(b)~planar networks whose terminal locations lie on the boundaries of at most
three designated faces, with each session's endpoints on one such face; and
(c)~networks with arbitrarily many nodes and terminal locations under a
structural restriction on session endpoints. We give a new proof that the
conjecture holds for networks with at most six coding nodes, without
computer-aided search, and show that if the conjecture holds on
\(\Gamma_{3,3}\), then it holds whenever no three sessions have six distinct
terminal locations.
\end{abstract}


\section{Introduction}
\label{sec:introduction}

Network coding allows intermediate nodes to encode received messages before
transmission~\cite{ahlswede2000network,koetter2003algebraic}. The
benefits and limitations of network coding are well understood in
unicast~\cite{sengupta2007analysis}, broadcast~\cite{fragouli2006broadcast},
and multicast~\cite{yin2013graphminor,yin2014matroid,braverman2025multicast}
communication. However, the multiple-unicast setting is far more intricate because an
intermediate node may encode messages from different sessions into a single
transmitted symbol. In directed multiple-unicast networks, coding can provide
a throughput advantage over fractional routing. Whether the advantage exists
in undirected multiple-unicast networks is one of the most important open
problems in network coding. Despite a large body of work over more than two decades, known results apply only to structurally restricted classes of
networks.

\begin{figure}[!b]
\centering
\resizebox{.9\columnwidth}{!}{%
\begin{tikzpicture}[
 >=Latex,
 terminal/.style={
   circle,
   draw,
   fill=white,
   inner sep=0pt,
   minimum size=6.4mm,
   font=\small
 },
 relay/.style={
   circle,
   draw,
   fill=white,
   inner sep=0pt,
   minimum size=5.2mm,
   font=\scriptsize
 },
 base/.style={
   draw=black!42,
   line width=.55pt
 },
 coding/.style={
   -{Latex[length=1.9mm]},
   draw=black,
   line width=.72pt,
   shorten <=1.4pt,
   shorten >=1.4pt
 },
 route1/.style={
   -{Latex[length=1.8mm]},
   draw=holecolor!85!black,
   line width=.90pt,
   shorten <=1.4pt,
   shorten >=1.4pt
 },
 route2/.style={
   -{Latex[length=1.8mm]},
   draw=teal!78!black,
   line width=.90pt,
   shorten <=1.4pt,
   shorten >=1.4pt
 },
 routeid/.style={
   fill=white,
   inner sep=.35pt,
   text=black,
   font=\scriptsize
 }
]

\begin{scope}
\node[font=\small] at (1.75,2.82) {(a) coding solution};

\node[terminal] (s1) at (0,2.15) {$s_1$};
\node[terminal] (s2) at (3.50,2.15) {$s_2$};
\node[terminal] (t2) at (0,0) {$t_2$};
\node[terminal] (t1) at (3.50,0) {$t_1$};
\node[relay] (u) at (1.75,1.62) {$u$};
\node[relay] (v) at (1.75,.53) {$v$};

\foreach \a/\b in
  {s1/u,s2/u,u/v,s1/t2,s2/t1,v/t2,v/t1}
  {\draw[base] (\a)--(\b);}

\draw[coding] (s1) to[bend left=10]
  node[pos=.48,above,font=\scriptsize] {$X_1$} (u);

\draw[coding] (s2) to[bend right=10]
  node[pos=.48,above,font=\scriptsize] {$X_2$} (u);

\draw[coding] (u) to[bend right=10]
  node[pos=.52,left,font=\scriptsize] {$X_1\!\oplus\!X_2$} (v);

\draw[coding] (s1) to[bend right=10]
  node[pos=.52,left,font=\scriptsize] {$X_1$} (t2);

\draw[coding] (s2) to[bend left=10]
  node[pos=.52,right,font=\scriptsize] {$X_2$} (t1);

\draw[coding] (v) to[bend left=12]
  node[pos=.53,below,font=\scriptsize] {$X_1\!\oplus\!X_2$} (t2);

\draw[coding] (v) to[bend right=12]
  node[pos=.53,below,font=\scriptsize] {$X_1\!\oplus\!X_2$} (t1);
\end{scope}

\begin{scope}[xshift=5.25cm]
\node[font=\small] at (1.75,2.82)
  {(b) fractional-routing solution};

\node[terminal] (S1) at (0,2.15) {$s_1$};
\node[terminal] (S2) at (3.50,2.15) {$s_2$};
\node[terminal] (T2) at (0,0) {$t_2$};
\node[terminal] (T1) at (3.50,0) {$t_1$};
\node[relay] (U) at (1.75,1.62) {$u$};
\node[relay] (V) at (1.75,.53) {$v$};

\foreach \a/\b in
  {S1/U,S2/U,U/V,S1/T2,S2/T1,V/T2,V/T1}
  {\draw[base] (\a)--(\b);}

\draw[route1] (S1) to[bend left=10]
  node[pos=.35,above,routeid] {$1$} (U);

\draw[route1] (U) to[bend left=10] (S2);

\draw[route1] (S2) to[bend left=8] (T1);

\draw[route1] (S1) to[bend right=8]
  node[pos=.38,left,routeid] {$1$} (T2);

\draw[route1] (T2) to[bend left=14] (V);

\draw[route1] (V) to[bend left=14] (T1);

\draw[route2] (S2) to[bend left=10]
  node[pos=.35,above,routeid] {$2$} (U);

\draw[route2] (U) to[bend left=10] (S1);

\draw[route2] (S1) to[bend left=11] (T2);

\draw[route2] (S2) to[bend right=11]
  node[pos=.38,right,routeid] {$2$} (T1);

\draw[route2] (T1) to[bend left=14] (V);

\draw[route2] (V) to[bend left=14] (T2);
\end{scope}

\end{tikzpicture}%
}

\caption{
A coding solution and a routing solution on the same undirected
unit-capacity network, adapted from~\cite{li2004network}.
Each session has rate one.
In~(a), node \(u\) sends the encoded symbol
\(X_1\!\oplus\!X_2\) to \(v\), which forwards it to both destinations.
In~(b), the two routes marked \(1\) connect \((s_1,t_1)\),
whereas the two routes marked \(2\) connect \((s_2,t_2)\).
Each route carries rate \(0.5\), so each session achieves rate one;
colors provide an additional visual cue.
}
\label{fig:coding-routing-undirected}
\end{figure}
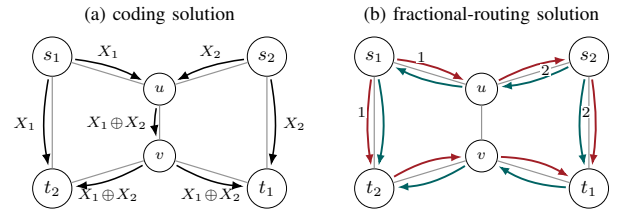

The coding solution and fractional-routing solution shown in
Fig.~\ref{fig:coding-routing-undirected} both attain unit rate for each
session, illustrating that coding need not provide a throughput advantage in an undirected network. Li and
Li~\cite{li2004network} and, independently,
Harvey~\etal~\cite{harvey2004comparing} formulated the multiple-unicast
conjecture in 2004. For an edge-capacitated network, the \emph{capacity
region} is the set of rate vectors that can be simultaneously achieved
under the edge capacities. The conjecture asserts that, for independent unicast
sessions in every undirected network, the network-coding and fractional-routing capacity regions coincide.

Resolving this open problem would have important implications for
computational complexity. A proof of
the conjecture would yield long-sought lower bounds for external-memory
integer sorting~\cite{farhadi2019sorting}, for constant-degree Boolean
circuits for multiplication~\cite{afshani2019lower}, and for certain
restricted data-structure problems~\cite{dvorak2021data}. If the conjecture
fails, Braverman~\etal~\cite{braverman2017coding} showed that any throughput
advantage, however small, can be amplified into a polylogarithmic gap between
coding and fractional routing on a family of undirected networks. These
consequences make the conjecture central to both network coding and
computational complexity.

One approach to establishing the multiple-unicast conjecture for a given
network is through \emph{cut sufficiency}, the property that the cut condition
is sufficient for fractional routing. In an undirected multiple-unicast
network, the \emph{cut condition} requires the capacity of every cut to be at
least the total rate of the sessions it separates, and every rate vector
achievable by a coding solution satisfies this condition. Consequently, under
cut sufficiency, every rate vector achievable by coding is fractionally
routable; since fractional routing is a special case of coding, the
network-coding and fractional-routing capacity regions coincide.

Known cut-sufficiency results prove that the conjecture holds for two
sessions~\cite{hu1963multi} and for networks with at most four terminal
locations~\cite{seymour1980four}. For a planar network, fix a planar embedding and
then select a set of faces; each selected face is a \emph{hole}. In the planar
networks considered here, every terminal location lies on the boundary of these holes, and the two endpoints of each session lie on the boundary of the same hole; different sessions may use different holes. The conjecture holds in the
one-hole case~\cite{okamura1981multicommodity} and in the two-hole
case~\cite{okamura1983flows}. The next natural cases are networks with five terminal locations and the
three-hole planar case. In these two cases, the cut condition need not be
sufficient for fractional routing~\cite{karzanov1990sums,karzanov1994metrics}.
Before this work, both cases were open; here we prove that the conjecture holds in each.

Failure of cut sufficiency does not imply that coding has a throughput
advantage. Jain~\etal~\cite{jain2006capacity} introduced the Input--Output and
Crypto inequalities and used them to prove the conjecture for their
four-session cyclic instance on the five-node graph \(K_{3,2}\), where cut
sufficiency fails. Thus, cut insufficiency already appears at five nodes, but
it does not produce a coding advantage. A broader result is known when the
number of coding nodes is bounded. A \emph{coding node} may encode or decode
messages; all other nodes only forward them. Yin~\etal~\cite{yin2018reduction}
proved the conjecture for all networks with at most six coding nodes using a
reduction approach and a computer-aided search. This remains the largest known
bound on the number of coding nodes. We give a new proof of the same conclusion, replacing computer-aided search
with the classification of six-point metrics.

Whereas the six-node result bounds the number of coding nodes, a different
structural direction allows arbitrarily many nodes and terminal locations
and restricts only the pattern of session endpoints. The demand graph \(D\) has one vertex for each terminal location and one
edge for each distinct unordered pair of session endpoints. Its matching
number \(\nu(D)\) is the largest number of edges no two of which share an
endpoint. Thus, \(\nu(D)\le2\) means that no three sessions have six distinct
terminal locations among their endpoints. The structural complexity of routing
changes sharply once the matching number exceeds two, even when all edge
capacities are integers~\cite[Sec.~1]{hirai2014maximum}. Consequently, the
condition \(\nu(D)\le2\) singles out a natural and important family of demand
graphs. However, cut sufficiency does not hold throughout this
family~\cite[Sec.~1]{hirai2010metric}. We prove the conjecture for a new class in this
family with arbitrarily many nodes and terminal locations, and show that if the
conjecture holds for every multiple-unicast instance on the fixed supply graph
\(\Gamma_{3,3}\), then it holds for every instance~in~the~family.

We develop a unified metric framework that underlies all of our results and
provides a reusable tool for studying further instances of the conjecture. Its
central perspective is that the basic objects behind the comparison between
coding and routing are not cuts alone, but graph metrics. The framework decomposes a network's shortest-path metric into a weighted combination of graph metrics obtained from a collection of fixed graphs. When one of these graphs is \(K_2\), the corresponding graph metrics are exactly cut metrics, making the cut-sufficiency approach a special case of the framework. We show that coding cannot outperform routing on each fixed graph. Our Transfer Lemma lifts these comparisons to any network with such a decomposition; the weighted decomposition then combines them to prove the conjecture for that network. Our contributions are as follows.

\begin{itemize}
  \item We develop a unified framework for systematically proving the
multiple-unicast conjecture for new network classes beyond cut sufficiency.

  \item We prove the conjecture for three new network classes:
    (a) undirected networks with at most five terminal locations;
    (b) planar networks with a fixed embedding whose terminal locations lie on
    the boundaries of at most three designated faces, with each session's
    endpoints on the same designated face; and (c) undirected networks with
    arbitrarily many nodes and terminal locations whose demand graph is a
    subgraph of the union of \(K_3\) and~a~star.

  \item We give a new proof that the conjecture holds for
  networks with at most six coding nodes.

  \item We show that, if the conjecture holds on \(\Gamma_{3,3}\), then
  it holds whenever no three sessions have six distinct terminal locations
  among their endpoints.
\end{itemize}

\section{Related Work}
\label{sec:related}

\subsection{Known Results from Cut Sufficiency}

Cut sufficiency gives several known cases of the multiple-unicast conjecture.
Every coding solution satisfies the cut condition; therefore, whenever the cut
condition is sufficient for fractional routing, the conjecture holds for that
instance. Hu~\cite{hu1963multi} proved this for two sessions, and
Seymour~\cite{seymour1980four} for instances with at most four terminal
locations.

Several planar results give further such cases. Seymour~\cite{seymour1981plane}
proved cut sufficiency when a network and its demand graph can be drawn
together in the plane without crossings. Okamura and Seymour~\cite{okamura1981multicommodity}
proved the one-hole case, and Okamura~\cite{okamura1983flows} proved the
two-hole case. For series-parallel networks, Chakrabarti~\etal~\cite{chakrabarti2012cut}
gave a complete characterization. The cut condition is sufficient for
fractional routing exactly when the network together with its demand graph
excludes an odd-spindle minor. In the five-terminal and three-hole cases, the
cut condition need not be sufficient for fractional routing, motivating the
graph-metric framework developed here.

\subsection{Metric Packing}

A graph metric is obtained by choosing a graph \(Q\), mapping each node of a
network to a node of \(Q\), and taking the shortest-path distance between the
images. When \(Q=K_2\), this gives a cut metric: nodes with the same image have
distance zero, and nodes with different images have distance one. Given
nonnegative lengths on the network edges, a metric packing is a nonnegative
combination of graph metrics whose value on every network edge is at most the
assigned length and whose value between the endpoints of each session equals
their shortest-path distance. Metric
packing concerns fractional routing; it does not itself compare coding with routing. Our framework developed here provides the new
connection that allows these known packings to be used in proofs of the multiple-unicast conjecture.

Karzanov~\cite{karzanov1990sums} showed that, when the sessions use at most
five terminal locations, the shortest-path distances between their endpoints
admit a packing by cut metrics and \(K_{2,3}\)-metrics. For a planar graph with
three holes, Karzanov~\cite{karzanov1994metrics} obtained a packing by the same
two types of metrics for terminal pairs whose endpoints lie on the boundary of
a common hole. Cut metrics alone do not suffice in general in either case.

Hirai~\cite{hirai2010metric} studied metric packing for the demand graph
\(K_3+K_3\), where \(\Gamma_{3,3}\)-metrics arise. Separately,
Grishukhin~\cite{grishukhin1992extreme} classified the extreme rays of the
metric cone on six points. Our unified framework establishes the coding--routing comparison for the graph
metrics that appear in these packing and classification results.

\subsection{Information-Theoretic Methods}

Jain~\etal~\cite{jain2006capacity} introduced the Input--Output and Crypto
inequalities and proved that the conjecture holds for their \(K_{3,2}\) cyclic
instance with four unicast sessions. Harvey~\etal~\cite{harvey2006capacity}
introduced informational dominance, a relation between edge sets, and combined
it with entropy inequalities to prove equality of the network coding and
fractional-routing rates for an infinite class of special bipartite instances. Yin~\etal~\cite{yin2018reduction} combined a
reduction approach with a computer-aided search to prove that the conjecture
holds for networks with at most six coding nodes. Qureshi and
Thakor~\cite{qureshi2022partition} later proved that the conjecture holds for
two infinite families of complete multipartite networks.

Our Transfer Lemma provides the connection between a unit-length graph case where the conjecture holds and the corresponding graph-metric comparison in our framework.

\subsection{Space Information Flow}

Li and Wu~\cite{liwu2012space} introduced space information flow, which
compares coding and routing costs in a metric space rather than on a fixed
network graph. Xiahou~\etal~\cite{xiahou2014geometric} used a single metric
embedding to compare a graph instance with a metric-space instance. If the
embedding preserves the shortest-path distance between the endpoints of every
session and does not increase the distance between the endpoints of any
network edge, equality of the coding and routing costs in the metric-space
instance implies the same equality in the graph instance.

Liu and Li~\cite{liu2025space} proved that the coding and routing costs are
equal in \(\ell_p^n\) for \(1\le p\le2\), using almost-isometric embeddings
of \(\ell_p^n\) into \(\ell_1\). Since every finite \(\ell_1\)
semimetric is a nonnegative combination of cut metrics~\cite{gupta2004cuts},
the final one-dimensional reduction of Xiahou~\etal~\cite{xiahou2014geometric}
can be viewed as a cut-metric packing. Our framework also permits non-cut
graph metrics: the \(K_{2,3}\)-metrics needed in the five-terminal and
three-hole cases~\cite{karzanov1990sums,karzanov1994metrics} cannot be
expressed as nonnegative combinations of cut metrics.

\subsection{Session Dominance}

Liu~\etal~\cite{liu2026dominance} introduced session dominance, which uses
shortest-path containment to compare collections of sessions. When a simpler
collection dominates another collection, a proof that the conjecture holds for
the dominating collection also applies to the other collection. Session dominance yields reductions for three-session instances and identifies
an infinite class specified by network topology, which neither contains nor is
contained in any of the network classes proved here.

\subsection{Strongly Reachable \texorpdfstring{$k$}{k}-Pair Networks}

A substantial complementary line of work studies the Langberg--M\'edard
multiple-unicast conjecture~\cite{langberg2009multiple}. It concerns directed acyclic $k$-pair networks that are strongly reachable: for every receiver $t_j$, there exist $k$ pairwise arc-disjoint directed paths, one from each source $s_i$ to $t_j$. The conjecture states that the underlying undirected network
admits a fractional multi-commodity flow of rate one for every pair.

Cai and Han~\cite{cai2022threepair} proved the conjecture for every strongly
reachable three-pair network. Cai and Han~\cite{cai2024backbone} also proved
it for strongly reachable networks supported on a collapsed backbone in the
form of a rooted binary tree. Taken together, these results establish strongly reachable \(k\)-pair networks
as a well-developed and structurally distinctive setting for studying when
fractional routing suffices for multiple-unicast communication. This line is complementary to our work: it imposes a directed reachability condition and derives an undirected routing conclusion, while our results impose conditions on the undirected instance through terminal locations, a fixed planar embedding, or the pattern of session endpoints.
\section{Preliminaries}\label{sec:prelim}

\subsection{Network Model}

Let \(G=(V,E)\) be a finite connected undirected graph, called the
\emph{supply graph}, and let \(c:E\to\mathbb{R}_{\ge 0}\) be an
edge-capacity assignment. Let \(\mathcal H=(h_1,\ldots,h_k)\) denote
\(k\) independent unicast sessions. Session \(h_i\) has source node
\(s_i\), receiver node \(t_i\ne s_i\), and rate \(r_i\). A rate vector is
\(\mathbf r=(r_1,\ldots,r_k)\).

Different sessions may have the same ordered pair of endpoints. Let
\(T:=\{s_i,t_i:1\le i\le k\}\) be the set of terminal locations, and let
\(\mathcal P(\mathcal H)\) be the set of distinct unordered terminal pairs
\(\{s_i,t_i\}\). The \emph{demand graph} is the simple
undirected graph \(D(\mathcal H):=(T,\mathcal P(\mathcal H))\). Its matching number
\(\nu(D(\mathcal H))\) is the maximum number of edges no two of which share
a terminal location. It records which unordered
terminal pairs occur among the sessions, but not their orientations, rates, or
multiplicities.

We use the standard zero-error network-coding model. In a coding solution,
each node may transmit functions of the source messages available at that node
and of messages received earlier, and each destination node must recover its
requested source message. For every supply edge \(e=uv\), let \(f(\overrightarrow{uv})\) and
\(f(\overrightarrow{vu})\) denote the average transmission rates of a coding
solution from \(u\) to \(v\) and from \(v\) to \(u\), respectively. The total
transmission rate on \(e\) is
\(f(e):=f(\overrightarrow{uv})+f(\overrightarrow{vu})\). The coding solution
achieves \(\mathbf r\) if it delivers the source message of session \(h_i\)
from \(s_i\) to \(t_i\) at rate \(r_i\) for every \(i\). The coding solution
is feasible under \(c\) if \(f(e)\le c(e)\) for every \(e\in E\).

Fractional routing does not encode messages from different sessions together.
A rate vector \(\mathbf r\) is achievable by fractional routing if
nonnegative rates can be assigned to the \(s_i\)--\(t_i\) paths of each
session \(h_i\) so that the rates assigned to these paths sum to \(r_i\) for
every \(i\), and, for every edge \(e\), the rates assigned to all paths
containing \(e\) sum to at most \(c(e)\).

The network-coding capacity region \(\mathcal C_{\NC}(G,c,\mathcal H)\) is
the set of all rate vectors that can be simultaneously achieved by coding
solutions under the edge capacities, and the routing capacity region
\(\mathcal C_{\MCF}(G,c,\mathcal H)\) is defined analogously for fractional
routing. Since routing is a special case of coding,
\(\mathcal C_{\MCF}(G,c,\mathcal H)\subseteq
\mathcal C_{\NC}(G,c,\mathcal H)\).

\subsection{Throughput Domain}

The multiple-unicast conjecture has the following throughput-domain
form~\cite{li2004network,harvey2004comparing}.

\begin{conjecture}[Throughput Domain]\label{conj:throughput}
For every supply graph \(G\), every edge-capacity assignment \(c\), and any
independent unicast sessions \(\mathcal H\),
\begin{equation}\label{eq:throughput-conjecture}
  \mathcal C_{\NC}(G,c,\mathcal H)
  =
  \mathcal C_{\MCF}(G,c,\mathcal H).
\end{equation}
\end{conjecture}

\subsection{Cost Domain}

The cost-domain form compares the transmission cost of a coding solution with
the minimum cost of fractional routing for the same rate vector. Let
\(\ell:E\to\mathbb{R}_{\ge 0}\) be a nonnegative edge-cost assignment, and
let \(d_{G,\ell}(u,v)\) denote the shortest-path distance between \(u\) and
\(v\) under \(\ell\). Under \(\ell\), the transmission cost of a coding
solution is \(\sum_{e\in E}\ell(e)f(e)\).

\begin{conjecture}[Cost Domain]\label{conj:cost}
For every supply graph \(G\), any independent unicast sessions
\(\mathcal H\), every nonnegative edge-cost assignment \(\ell\), and every
coding solution that achieves \(\mathbf r\),
\begin{equation}\label{eq:cost-conjecture}
  \sum_{e\in E}\ell(e)f(e)
  \ge
  \sum_{i=1}^k r_i d_{G,\ell}(s_i,t_i).
\end{equation}
\end{conjecture}

The right-hand side of~\eqref{eq:cost-conjecture} is the minimum cost of
fractional routing that delivers \(\mathbf r\).

Li and Li~\cite{li2004network} established the following equivalence.

\begin{theorem}\label{thm:cost-equivalence}
Conjectures~\ref{conj:throughput} and~\ref{conj:cost} are equivalent.
\end{theorem}

\begin{proof}
Assume Conjecture~\ref{conj:throughput} and assign capacity \(f(e)\) to every
supply edge \(e\). The coding solution remains feasible, so \(\mathbf r\) is
fractionally routable under these capacities; its routing cost is at most the
left-hand side of~\eqref{eq:cost-conjecture}, while every routing of
\(\mathbf r\) has cost at least the right-hand side. Conversely, suppose that a coding solution feasible under some capacity
assignment \(c\) achieves \(\mathbf r\), but \(\mathbf r\) is not
fractionally routable under \(c\). Therefore, the dual of the fractional-routing linear program
supplies a nonnegative \(\ell\) for which the right-hand side of
\eqref{eq:cost-conjecture} exceeds \(\sum_{e\in E}\ell(e)c(e)\). Since
\(f(e)\le c(e)\) for every \(e\), this contradicts
Conjecture~\ref{conj:cost}.
\end{proof}

\subsection{Graph Metrics and Metric Packings}

A \emph{semimetric} on \(V\) is a nonnegative symmetric function
\(m:V\times V\to\mathbb{R}_{\ge 0}\) with \(m(u,u)=0\) that satisfies
\(m(u,v)\le m(u,w)+m(w,v)\) for all \(u,v,w\in V\). 

Let \(Q\) be a connected graph whose edges have unit length, and let \(d_Q\)
denote its shortest-path distance.

\begin{definition}[\(Q\)-metric]\label{def:qmetric}
Let \(\phi:V(G)\to V(Q)\) be a map. The \(Q\)-metric induced by \(\phi\) is
\begin{equation}\label{eq:qmetric}
  m_{Q,\phi}(u,v)
  :=
  d_Q\bigl(\phi(u),\phi(v)\bigr).
\end{equation}
The map \(\phi\) may identify several nodes of \(G\).
\end{definition}

A \(K_2\)-metric is a cut metric: it assigns distance one to pairs separated
by the cut induced by the map and zero to all other pairs. A
\(K_{2,3}\)-metric is a \(Q\)-metric with \(Q=K_{2,3}\).

\begin{definition}[Fractional metric packing]\label{def:packing}
Let \(R\subseteq\binom{V}{2}\) be a set of terminal pairs. A fractional
metric packing for \((G,\ell,R)\) consists of \(Q_a\)-metrics
\(m_a=m_{Q_a,\phi_a}\) and weights \(\lambda_a\ge 0\) such that
\begin{align}
  \sum_a \lambda_a m_a(u,v)
  &\le \ell(uv)
  && \text{for every } uv\in E,
  \label{eq:packing-edge}\\
  \sum_a \lambda_a m_a(s,t)
  &= d_{G,\ell}(s,t)
  && \text{for every } \{s,t\}\in R.
  \label{eq:packing-pair}
\end{align}
\end{definition}

The first condition bounds the weighted metric sum by the cost of every supply
edge. The second condition matches the shortest-path distance of every pair in
\(R\).

For the planar results, every planar graph is equipped with a fixed planar
embedding. Fix a distinguished set \(\holes\) of faces of this embedding,
and call each face in \(\holes\) a \emph{hole}.

We use the following known metric-packing results.

\begin{theorem}[Known metric-packing results]
\label{thm:karzanov-packings}
Let \(G\) have a nonnegative edge-cost assignment \(\ell\), and let
\(R\subseteq\binom{V}{2}\) be a set of terminal pairs. Define
\(T_R:=\bigcup_{\{u,v\}\in R}\{u,v\}\) and \(D_R:=(T_R,R)\).
\begin{enumerate}
  \item If \(\lvert T_R\rvert\le5\), then there is a fractional metric
  packing for \((G,\ell,R)\) by \(K_2\)- and
  \(K_{2,3}\)-metrics~\cite[Sec.~3]{karzanov1990sums}.

  \item If \(G\) is a three-hole planar graph and the two endpoints of every
  pair in \(R\) lie on the boundary of a common hole, then there is a
  fractional metric packing for \((G,\ell,R)\) by \(K_2\)- and
  \(K_{2,3}\)-metrics~\cite[pp.~118--119]{karzanov2019packing}.

  \item If \(D_R\) is a subgraph of the union of two stars, then there is a
  fractional metric packing for \((G,\ell,R)\) by
  \(K_2\)-metrics~\cite[(1.2)]{karzanov1990sums}.

  \item If \(D_R\) is a subgraph of the union of a triangle and a
  star, then there is a fractional metric packing for \((G,\ell,R)\) by
  \(K_2\)- and \(K_{2,3}\)-metrics~
  \cite[Sec.~3]{karzanov1990sums}.
\end{enumerate}
\end{theorem}

\begin{figure*}[!t]
\centering
\begingroup
\definecolor{transferred}{RGB}{165,15,15}
\resizebox{.97\textwidth}{!}{%
\begin{tikzpicture}[
  x=.82cm,
  y=.86cm,
  line cap=round,
  line join=round,
  gedge/.style={draw=black!85,line width=.66pt},
  rededge/.style={draw=transferred!92!black,line width=1.00pt},
  gedgedashed/.style={
    draw=black!78,densely dashed,line width=.7pt},
  qpath/.style={draw=transferred!92!black,line width=1.10pt},
  qfinal/.style={draw=black!85,line width=.86pt},
  processarrow/.style={draw=black!58,line width=.64pt,->},
  vtx/.style={
    circle,draw=black!85,fill=white,
    line width=.54pt,minimum size=4.15mm,inner sep=0pt,font=\footnotesize
  },
  panel/.style={font=\footnotesize\bfseries},
  note/.style={font=\scriptsize,align=center}
]

\node[panel] at (2.10,4.10) {(a) Original \(G\)};
\coordinate (aaone) at (.40,2.55);
\coordinate (aatwo) at (.40,.95);
\coordinate (bbone) at (3.80,2.55);
\coordinate (bbtwo) at (3.80,.95);
\coordinate (ppone) at (2.10,3.22);
\coordinate (pptwo) at (2.10,1.85);
\coordinate (ppthree) at (2.10,.48);
\draw[gedge] (aaone)--(bbone);
\draw[gedge] (aatwo)--(bbtwo);
\draw[rededge] (ppone) to[out=0,in=0,looseness=1.10] (ppthree);
\draw[gedge] (bbone)--(ppone);
\draw[gedge] (bbtwo)--(ppthree);
\draw[gedge] (aaone)--(pptwo);
\draw[gedge] (bbone)--(pptwo);
\node[vtx] at (aaone)   {\(a_1\)};
\node[vtx] at (aatwo)   {\(a_2\)};
\node[vtx] at (bbone)   {\(b_1\)};
\node[vtx] at (bbtwo)   {\(b_2\)};
\node[vtx] at (ppone)   {\(p_1\)};
\node[vtx] at (pptwo)   {\(p_2\)};
\node[vtx] at (ppthree) {\(p_3\)};
\node[note,text=transferred!92!black] at (3.35,1.85) {\(e\)};

\draw[processarrow] (4.15,1.85)--(4.85,1.85);
\node[note] at (4.50,2.28) {\(\phi\)};

\node[panel] at (6.85,4.10) {(b) Identify nodes};
\coordinate (a2)  at (5.25,1.85);
\coordinate (b2)  at (8.45,1.85);
\coordinate (p12) at (6.85,3.05);
\coordinate (p22) at (6.85,1.85);
\coordinate (p32) at (6.85,.65);
\draw[gedge] (a2) to[bend left=46] (b2);
\draw[gedge] (a2) to[bend right=46] (b2);
\draw[rededge] (p12) to[out=0,in=0,looseness=1.12] (p32);
\draw[gedge] (b2) to[bend right=12] (p12);
\draw[gedge] (b2) to[bend left=12] (p32);
\draw[gedge] (a2)--(p22)--(b2);
\node[vtx] at (a2)  {\(a\)};
\node[vtx] at (b2)  {\(b\)};
\node[vtx] at (p12) {\(p_1\)};
\node[vtx] at (p22) {\(p_2\)};
\node[vtx] at (p32) {\(p_3\)};

\draw[processarrow] (8.75,1.85)--(9.55,1.85);

\node[panel] at (11.60,4.10) {(c) Replace \(e\) by \(P_e\)};
\coordinate (a3)  at (10.00,1.85);
\coordinate (b3)  at (13.20,1.85);
\coordinate (p13) at (11.60,3.05);
\coordinate (p23) at (11.60,1.85);
\coordinate (p33) at (11.60,.65);
\draw[gedge] (a3) to[bend left=46] (b3);
\draw[gedge] (a3) to[bend right=46] (b3);
\draw[gedgedashed] (p13) to[out=0,in=0,looseness=1.12] (p33);
\draw[gedge] (b3) to[bend right=12] (p13);
\draw[gedge] (b3) to[bend left=12] (p33);
\draw[gedge] (a3)--(p23)--(b3);
\draw[qpath] (p13)--(a3)--(p33);
\node[vtx] at (a3)  {\(a\)};
\node[vtx] at (b3)  {\(b\)};
\node[vtx] at (p13) {\(p_1\)};
\node[vtx] at (p23) {\(p_2\)};
\node[vtx] at (p33) {\(p_3\)};
\node[note,text=transferred!92!black] at (10.46,3.02) {\(P_e\)};

\draw[processarrow] (13.50,1.85)--(14.25,1.85);
\node[note] at (13.88,2.28) {\(\forall e\)};

\node[panel] at (16.35,4.10) {(d) \(Q=K_{2,3}\)};
\coordinate (a4)  at (14.75,1.85);
\coordinate (b4)  at (17.95,1.85);
\coordinate (p14) at (16.35,3.05);
\coordinate (p24) at (16.35,1.85);
\coordinate (p34) at (16.35,.65);
\draw[qfinal] (a4)--(p14);
\draw[qfinal] (p14)--(b4);
\draw[qfinal] (a4)--(p24);
\draw[qfinal] (p24)--(b4);
\draw[qfinal] (a4)--(p34);
\draw[qfinal] (p34)--(b4);
\node[vtx] at (a4)  {\(a\)};
\node[vtx] at (b4)  {\(b\)};
\node[vtx] at (p14) {\(p_1\)};
\node[vtx] at (p24) {\(p_2\)};
\node[vtx] at (p34) {\(p_3\)};

\end{tikzpicture}%
}
\endgroup
\caption{Construction behind the Transfer Lemma with \(Q=K_{2,3}\). In (a),
the map \(\phi\) sends \(a_1,a_2\) to \(a\), \(b_1,b_2\) to \(b\), and each
\(p_j\) to itself. The highlighted curve marks one representative supply edge
\(e\). In (b), nodes with the same image under \(\phi\) are identified;
parallel curves represent parallel edges created by this identification. Panel
(c) shows one replacement: the dashed black curve is \(e\) before replacement,
and the solid dark-red path \(P_e\), consisting of edges of \(Q\), replaces it.
After every supply edge is replaced in this way, panel (d) shows the resulting
graph \(Q\).}
\label{fig:transfer-lemma}
\end{figure*}
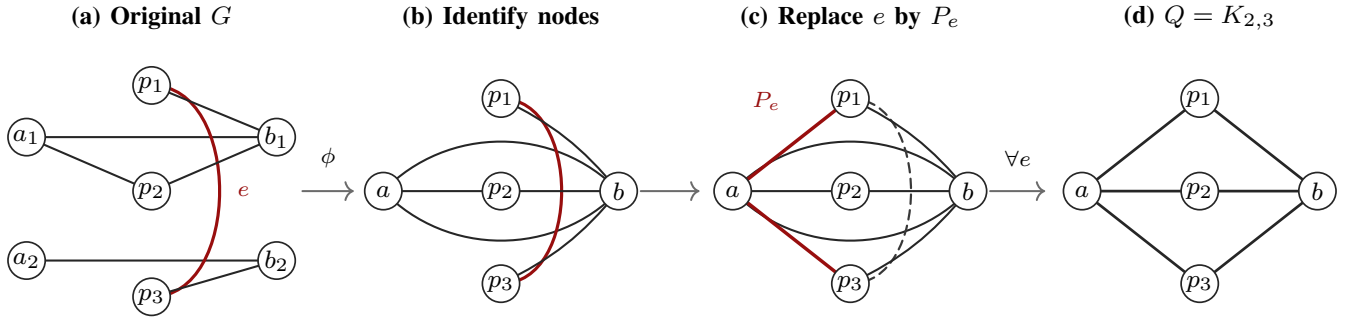

\section{A Unified Metric Framework}\label{sec:framework}

Our unified metric framework takes the perspective that the basic objects
behind the comparison between coding and routing are not cuts alone, but graph
metrics. For each graph metric in a metric packing, we establish on the
associated unit-length graph \(Q\) that each coding solution has transmission
cost at least the minimum fractional-routing cost for the same rate vector.
Our Transfer Lemma transfers this inequality to the induced \(Q\)-metric on
the supply graph; the metric packing then combines the resulting graph-metric
inequalities.

\subsection{Graph-Metric Inequalities}

\begin{definition}[Graph-metric inequality]
\label{def:metric-inequality}
Let $m=m_{Q,\phi}$ be a $Q$-metric on $V(G)$. We say that $m$ satisfies the
\emph{graph-metric inequality} for $\mathcal H$ if every coding solution on
$G$ that achieves $\mathbf r$ satisfies
\begin{equation}\label{eq:metric-inequality}
  \sum_{e=uv\in E} m(u,v)f(e)
  \ge
  \sum_{i=1}^{k} r_i m(s_i,t_i).
\end{equation}
\end{definition}

The left-hand side of~\eqref{eq:metric-inequality} weights the total
transmission rate on each supply edge by its distance under $m$. The
right-hand side is the corresponding weighted sum over the session pairs.

Yin~\etal~\cite[Corollary~1]{yin2018reduction} proved that, if $Q$ is a
connected bipartite graph with unit edge lengths and at most one distinct
ordered endpoint pair among the sessions has distance greater than two, then
every coding solution on $Q$ satisfies
\begin{equation}\label{eq:yin-unit}
  \sum_{xy\in E(Q)} f(xy)
  \ge
  \sum_{i=1}^{k} r_i d_Q(s_i,t_i).
\end{equation}

Our extension permits sessions of distance greater than two in both
orientations of one unordered endpoint pair.

\begin{lemma}[Bidirectional long-pair extension]
\label{lem:bidirectional}
Let $Q$ be a connected bipartite graph with unit edge lengths. Suppose that
there are nodes $a,b\in V(Q)$ with $d_Q(a,b)>2$ such that every distinct
ordered endpoint pair of source--receiver distance greater than two is
$(a,b)$ or $(b,a)$. Then every coding solution on $Q$ satisfies
\begin{equation}\label{eq:bidirectional-unit}
  \sum_{xy\in E(Q)} f(xy)
  \ge
  \sum_{i=1}^{k} r_i d_Q(s_i,t_i).
\end{equation}
\end{lemma}

\begin{proof}
Sessions with the same ordered endpoints can be merged into one session: its
message consists of the original independent messages, so its rate is their
sum. If the remaining sessions of distance greater than two have only
one orientation,~\eqref{eq:yin-unit} applies. Otherwise, one source message
has ordered endpoints $(a,b)$ and the other has ordered endpoints $(b,a)$.
Appendix~\ref{app:bidirectional} extends the layer-counting argument to show
that both source messages are counted with the required multiplicity, which
gives~\eqref{eq:bidirectional-unit}.
\end{proof}

Our Transfer Lemma transfers the coding--routing comparison in
\eqref{eq:transfer-assumption} on \(Q\) to the graph-metric inequality
\eqref{eq:metric-inequality} on \(G\). The proof uses an edge-to-path
replacement result established by Langberg and
Effros~\cite[Definition~2 and Lemma~3]{langberg2021edgeremoval}. Let \(e=uv\)
be an edge of capacity \(\lambda\), and let \(P\) be a \(u\)--\(v\) path in
\(G-e\). Let \(G'\) be obtained from \(G\) by deleting \(e\) and increasing
the capacity of every edge of \(P\) by \(\lambda\). If a zero-error coding
solution on \(G\) achieves \(\mathbf r\), then, for every
\(\varepsilon>0\), there is a zero-error coding solution on \(G'\) for the
same sessions with rate vector \(\mathbf r'\) satisfying
\(\lvert r_i'-r_i\rvert<\varepsilon\) for every \(i\).

Fig.~\ref{fig:transfer-lemma} illustrates the node identification and
edge-to-path replacement used in the proof.
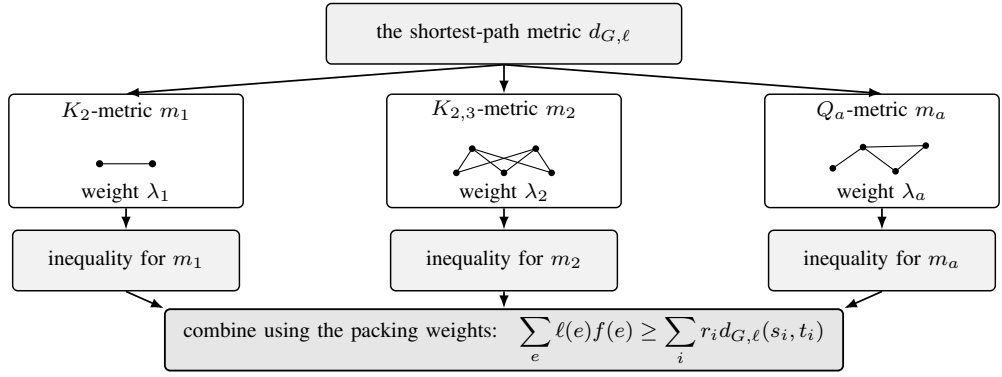
\begin{figure*}[!ht]
\centering
\begin{tikzpicture}[
  font=\footnotesize,
  >=Latex,
  box/.style={
    draw=black,
    rounded corners=2pt,
    align=center,
    minimum height=8mm,
    inner xsep=5pt,
    line width=.55pt
  },
  metricitem/.style={
    draw=black,
    rounded corners=2pt,
    minimum width=31mm,
    minimum height=15mm,
    align=center,
    fill=white,
    line width=.55pt
  },
  finalbox/.style={
    draw=black,
    rounded corners=2pt,
    align=center,
    minimum height=7mm,
    inner xsep=7pt,
    inner ysep=2pt,
    fill=black!10,
    line width=.70pt
  },
  arr/.style={
    -{Latex[length=1.7mm]},
    draw=black,
    line width=.65pt
  }
]

\node[
  box,
  fill=black!6,
  minimum width=47mm
] (metric) at (0,0)
{the shortest-path metric \(d_{G,\ell}\)};

\node[metricitem] (a1) at (-5.0,-1.55) {};
\node at (-5.0,-1.02) {\(K_2\)-metric \(m_1\)};
\draw (-5.35,-1.72)--(-4.65,-1.72);
\fill (-5.35,-1.72) circle (1.35pt);
\fill (-4.65,-1.72) circle (1.35pt);
\node at (-5.0,-2.1) {weight \(\lambda_1\)};

\node[metricitem] (a2) at (0,-1.55) {};
\node at (0,-1.02) {\(K_{2,3}\)-metric \(m_2\)};

\coordinate (L1) at (-.42,-1.52);
\coordinate (L2) at (.42,-1.52);
\coordinate (R1) at (-.63,-1.84);
\coordinate (R2) at (0,-1.84);
\coordinate (R3) at (.63,-1.84);

\foreach \x in {L1,L2}{
  \foreach \y in {R1,R2,R3}{
    \draw (\x)--(\y);
  }
}
\foreach \x in {L1,L2,R1,R2,R3}{
  \fill (\x) circle (1.25pt);
}
\node at (0,-2.1) {weight \(\lambda_2\)};

\node[metricitem] (a3) at (5.0,-1.55) {};
\node at (5.0,-1.02) {\(Q_a\)-metric \(m_a\)};

\draw
  (4.35,-1.78)--(4.75,-1.50)--(5.18,-1.82)--(5.58,-1.48);
\draw (4.75,-1.50)--(5.58,-1.48);

\foreach \p in {
  (4.35,-1.78),
  (4.75,-1.50),
  (5.18,-1.82),
  (5.58,-1.48)
}{
  \fill \p circle (1.25pt);
}
\node at (5.0,-2.1) {weight \(\lambda_a\)};

\draw[arr] (metric.south)--(a1.north);
\draw[arr] (metric.south)--(a2.north);
\draw[arr] (metric.south)--(a3.north);

\node[
  box,
  fill=black!5,
  minimum width=30mm
] (i1) at (-5.0,-3.00)
{inequality for \(m_1\)};

\node[
  box,
  fill=black!5,
  minimum width=30mm
] (i2) at (0,-3.00)
{inequality for \(m_2\)};

\node[
  box,
  fill=black!5,
  minimum width=30mm
] (i3) at (5.0,-3.00)
{inequality for \(m_a\)};

\draw[arr] (a1.south)--(i1.north);
\draw[arr] (a2.south)--(i2.north);
\draw[arr] (a3.south)--(i3.north);

\node[finalbox] (global) at (0,-4.05)
{combine using the packing weights:\quad
 \(\displaystyle
   \sum_e \ell(e)f(e)
   \ge
   \sum_i r_i d_{G,\ell}(s_i,t_i)\)};

\draw[arr] (i1.south)--(global.north west);
\draw[arr] (i2.south)--(global.north);
\draw[arr] (i3.south)--(global.north east);

\end{tikzpicture}
\caption{The unified metric framework. A metric packing supplies graph
metrics and weights. The constraints in the metric packing, together with the
graph-metric inequalities, show that every coding solution has transmission
cost at least the minimum fractional-routing cost for the same rate vector.}
\label{fig:framework}
\end{figure*}
\begin{lemma}[Transfer Lemma]\label{lem:transfer}
Let \(Q\) be a connected graph with unit edge lengths, let
\(\phi:V(G)\to V(Q)\) be a map, and let \(m=m_{Q,\phi}\). Suppose that, for
every rate vector \(\mathbf r\), every coding solution on \(Q\) for the
sessions \((\phi(s_i),\phi(t_i),r_i)\) satisfies
\begin{equation}\label{eq:transfer-assumption}
  \sum_{xy\in E(Q)} f_Q(xy)
  \ge
  \sum_{i=1}^{k} r_i
  d_Q\bigl(\phi(s_i),\phi(t_i)\bigr),
\end{equation}
where \(f_Q(xy)\) is the total transmission rate on edge \(xy\). Then \(m\)
satisfies the graph-metric inequality~\eqref{eq:metric-inequality} on \(G\).
\end{lemma}

\begin{proof}
Fix a coding solution on \(G\) that achieves \(\mathbf r\), with total
transmission rate \(f(e)\) on each supply edge \(e\). For every \(e=uv\)
with \(\phi(u)\ne\phi(v)\), choose a shortest
\(\phi(u)\)--\(\phi(v)\) path \(P_e\) in \(Q\).

Start with a copy of \(Q\) whose edges have capacity zero, and let each node
\(x\in V(Q)\) simulate all nodes in \(\phi^{-1}(x)\). For every such supply
edge \(e=uv\), add a distinct auxiliary edge \(a_e\) of capacity \(f(e)\)
between \(\phi(u)\) and \(\phi(v)\), carrying the same transmissions as \(e\)
in both directions. Transmissions on an edge \(e=uv\) with
\(\phi(u)=\phi(v)\) are local at \(\phi(u)\). Hence the resulting multigraph
has a coding solution that achieves \(\mathbf r\) for the sessions
\((\phi(s_i),\phi(t_i),r_i)\).

Let \(M\) be the number of auxiliary edges, and fix an order for replacing
them. Because \(P_e\) uses only edges of the copy of \(Q\), \(P_e\) is a
\(\phi(u)\)--\(\phi(v)\) path in the current multigraph minus \(a_e\) when
\(a_e\) is replaced. The edge-to-path replacement result therefore applies
successively. After all replacements, only the copy of \(Q\) remains, with
capacity \(c_Q(xy):=\sum_{e:\,xy\in E(P_e)}f(e)\) on each
\(xy\in E(Q)\).

If \(M=0\), the construction already gives a zero-error coding solution on
\(Q\) under \(c_Q\), and we set \(\mathbf q^{(n)}=\mathbf r\). If \(M>0\),
then, for each integer \(n\ge1\), apply the edge-to-path replacement result
at each replacement with coordinatewise approximation error less than
\(1/(nM)\). The resulting zero-error coding solution on \(Q\), feasible under
\(c_Q\), has a rate vector \(\mathbf q^{(n)}\) satisfying
\(\lvert q_i^{(n)}-r_i\rvert<1/n\) for \(1\le i\le k\) by the triangle
inequality. Hence \(\mathbf q^{(n)}\to\mathbf r\) in either case.

For every \(n\), applying~\eqref{eq:transfer-assumption} to the coding
solution with rate vector \(\mathbf q^{(n)}\) and then using the edge-capacity
constraints gives
\begin{equation}\label{eq:transfer-approximation-bound}
  \sum_{i=1}^{k} q_i^{(n)}
  d_Q\bigl(\phi(s_i),\phi(t_i)\bigr)
  \le
  \sum_{xy\in E(Q)} c_Q(xy).
\end{equation}
Letting \(n\to\infty\), using the definition of \(c_Q\), and recalling that
each \(P_e\) is shortest, we obtain
\begin{align}
  \sum_{i=1}^{k} r_i m(s_i,t_i)
  &\le
  \sum_{xy\in E(Q)} c_Q(xy)
  =
  \sum_{e=uv\in E} f(e)m(u,v).
  \label{eq:transfer-metric-bound}
\end{align}
Thus \(m\) satisfies~\eqref{eq:metric-inequality}.
\end{proof}

\subsection{Divide and Combine}

The final step uses the constraints in a fractional metric packing to combine
the graph-metric inequalities, as illustrated in
Fig.~\ref{fig:framework}.

\begin{theorem}[Divide-and-combine criterion]
\label{thm:divide-combine}
Fix a multiple-unicast instance \((G,\mathcal H)\). Suppose that, for every
nonnegative edge-cost assignment \(\ell\), there is a fractional metric
packing for \((G,\ell,\mathcal P(\mathcal H))\) such that every graph metric
in the metric packing satisfies~\eqref{eq:metric-inequality}. Then the
multiple-unicast conjecture holds for \((G,\mathcal H)\).
\end{theorem}

\begin{proof}
Fix \(\ell\) and a coding solution that achieves \(\mathbf r\). Let
\((m_a,\lambda_a)\) be the graph metrics and weights in the metric packing.
By~\eqref{eq:packing-edge}, \eqref{eq:metric-inequality}, and
\eqref{eq:packing-pair},
\begin{align}
  \sum_{e\in E}\ell(e)f(e)
  &\ge \sum_a\lambda_a \sum_{e=uv\in E}m_a(u,v)f(e),
  \label{eq:divide-combine-edge}\\
  &\ge \sum_a\lambda_a \sum_{i=1}^{k}r_i m_a(s_i,t_i),
  \label{eq:divide-combine-metric}\\
  &= \sum_{i=1}^{k}r_i d_{G,\ell}(s_i,t_i).
  \label{eq:divide-combine-distance}
\end{align}
The displayed inequality holds for every \(\ell\) and every coding solution.
Theorem~\ref{thm:cost-equivalence} therefore proves that the multiple-unicast
conjecture holds for \((G,\mathcal H)\).
\end{proof}

\begin{corollary}[\(K_2\)--\(K_{2,3}\) packing]
\label{cor:k2-k23-packing}
Suppose that, for every nonnegative edge-cost assignment \(\ell\), the pairs
in \(\mathcal P(\mathcal H)\) admit a fractional metric packing by
\(K_2\)- and \(K_{2,3}\)-metrics. Then the multiple-unicast conjecture holds
for \((G,\mathcal H)\).
\end{corollary}

\begin{proof}
Both \(K_2\) and \(K_{2,3}\) are connected bipartite graphs of diameter at
most two. Hence~\eqref{eq:yin-unit} applies to every collection of sessions
on either graph. The Transfer Lemma gives the required
graph-metric inequalities, and Theorem~\ref{thm:divide-combine} completes the
proof.
\end{proof}

\section{New Network Classes and a New Proof}
\label{sec:applications}

\subsection{Five Terminal Locations}

\begin{theorem}\label{thm:five}
The multiple-unicast conjecture holds for every network whose sessions use at
most five terminal locations.
\end{theorem}

\begin{proof}
Theorem~\ref{thm:karzanov-packings} provides the required
\(K_2\)--\(K_{2,3}\) fractional metric packing. Corollary~
\ref{cor:k2-k23-packing} proves the claim.
\end{proof}

The supply graph may have an arbitrary size and topology. The only restriction is
that the instance has at most five terminal locations; in particular,
its demand graph has at most \(\binom{5}{2}\) edges.

\subsection{Three-Hole Planar Networks}

\begin{theorem}\label{thm:three-holes}
Let \(G\) be a planar supply graph with a fixed planar embedding, and let
\(\holes\) be a set of at most three holes. If, for every session, both
terminal locations lie on the boundary of the same hole in \(\holes\), then
the multiple-unicast conjecture holds for the instance.
\end{theorem}

\begin{proof}
Theorem~\ref{thm:karzanov-packings} provides the required
\(K_2\)--\(K_{2,3}\) fractional metric packing. Corollary~
\ref{cor:k2-k23-packing} completes the proof.
\end{proof}

Fig.~\ref{fig:three-holes} illustrates such a network. The theorem imposes no
fixed bound on the number of terminal locations or sessions. Different sessions
may use different holes, but both terminal locations of each session must lie
on the boundary of the same hole.

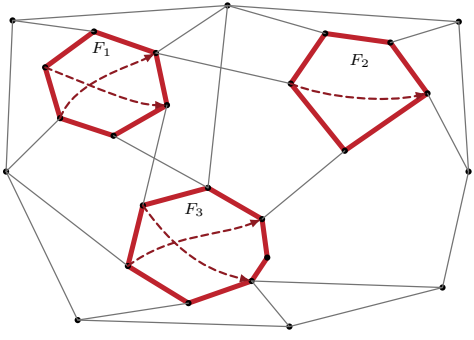
\begin{figure}[!ht]
\centering
\resizebox{.7\columnwidth}{!}{%
\begin{tikzpicture}[>=Latex,line cap=round,line join=round,
  supply/.style={draw=black!55,line width=.55pt},
  boundary/.style={draw=holecolor,line width=2.2pt},
  session/.style={
    draw=holecolor!75!black,
    densely dashed,
    -{Latex[length=1.8mm]},
    line width=.95pt
  },
  vnode/.style={circle,fill=black,inner sep=0pt,minimum size=2.7pt}
]

\coordinate (a1) at (-3.05,1.50);
\coordinate (a2) at (-2.30,2.05);
\coordinate (a3) at (-1.35,1.72);
\coordinate (a4) at (-1.18,.92);
\coordinate (a5) at (-2.00,.45);
\coordinate (a6) at (-2.82,.72);
\draw[boundary] (a1)--(a2)--(a3)--(a4)--(a5)--(a6)--cycle;

\coordinate (b1) at (.72,1.25);
\coordinate (b2) at (1.25,2.02);
\coordinate (b3) at (2.25,1.88);
\coordinate (b4) at (2.82,1.10);
\coordinate (b5) at (1.55,.22);
\draw[boundary] (b1)--(b2)--(b3)--(b4)--(b5)--cycle;

\coordinate (c1) at (-1.55,-.62);
\coordinate (c2) at (-.55,-.35);
\coordinate (c3) at (.28,-.83);
\coordinate (c4) at (.36,-1.42);
\coordinate (c5) at (.12,-1.78);
\coordinate (c6) at (-.85,-2.12);
\coordinate (c7) at (-1.78,-1.55);
\draw[boundary] (c1)--(c2)--(c3)--(c4)--(c5)--(c6)--(c7)--cycle;

\coordinate (p1) at (-3.55,2.22);
\coordinate (p2) at (-.25,2.45);
\coordinate (p3) at (3.30,2.20);
\coordinate (p4) at (-3.65,-.10);
\coordinate (p5) at (3.45,-.10);
\coordinate (p6) at (-2.55,-2.38);
\coordinate (p7) at (.70,-2.48);
\coordinate (p8) at (3.05,-1.88);

\node[vnode] at (p1) {};
\node[vnode] at (p2) {};
\node[vnode] at (p3) {};
\node[vnode] at (p4) {};
\node[vnode] at (p5) {};
\node[vnode] at (p6) {};
\node[vnode] at (p7) {};
\node[vnode] at (p8) {};

\node[vnode] at (a1) {};
\node[vnode] at (a2) {};
\node[vnode] at (a3) {};
\node[vnode] at (a4) {};
\node[vnode] at (a5) {};
\node[vnode] at (a6) {};

\node[vnode] at (b1) {};
\node[vnode] at (b2) {};
\node[vnode] at (b3) {};
\node[vnode] at (b4) {};
\node[vnode] at (b5) {};

\node[vnode] at (c1) {};
\node[vnode] at (c2) {};
\node[vnode] at (c3) {};
\node[vnode] at (c4) {};
\node[vnode] at (c5) {};
\node[vnode] at (c6) {};
\node[vnode] at (c7) {};

\draw[supply] (p1)--(p2)--(p3);
\draw[supply] (p1)--(p4)--(p6)--(p7)--(p8)--(p5)--(p3);

\draw[supply] (p2)--(a3);
\draw[supply] (p1)--(a2);
\draw[supply] (p4)--(a6);
\draw[supply] (a4)--(c1);

\draw[supply] (p2)--(b2);
\draw[supply] (p3)--(b3);
\draw[supply] (p5)--(b4);

\draw[supply] (p4)--(c7);
\draw[supply] (p6)--(c6);
\draw[supply] (p7)--(c5);
\draw[supply] (p8)--(c5);

\draw[supply] (a3)--(b1);
\draw[supply] (a5)--(c2);
\draw[supply] (b5)--(c3);
\draw[supply] (p2)--(c2);

\draw[session] (a6) .. controls (-2.62,1.28) and (-1.72,1.52) .. (a3);
\draw[session] (a1) .. controls (-2.18,1.15) and (-1.66,.92) .. (a4);
\draw[session] (b1) .. controls (1.28,1.02) and (2.20,.96) .. (b4);
\draw[session] (c7) .. controls (-1.12,-1.05) and (-.38,-1.10) .. (c3);
\draw[session] (c1) .. controls (-.94,-1.48) and (-.42,-1.65) .. (c5);

\node[text=black,font=\scriptsize\bfseries] at (-2.18,1.78) {$F_1$};
\node[text=black,font=\scriptsize\bfseries] at (1.78,1.60) {$F_2$};
\node[text=black,font=\scriptsize\bfseries] at (-.76,-.68) {$F_3$};

\end{tikzpicture}%
}
\caption{A three-hole planar supply graph. The highlighted face boundaries
are the holes; \(F_1\), \(F_2\), and \(F_3\) have six, five, and seven
boundary nodes, respectively. Each displayed session has both terminal
locations on the boundary of one hole; different sessions may use different
holes.}
\label{fig:three-holes}
\end{figure}

\subsection{The Union of \(K_3\) and a Star}

Fig.~\ref{fig:matching-two-patterns}(c) illustrates the union of a triangle and a star pattern. The demand graph is a subgraph of the 
union of a triangle and a star. The two graphs in this union need not be vertex-disjoint. The supply graph may have arbitrary size and
topology, and there is no bound on the number of terminal locations; only the distinct unordered endpoint pairs of the sessions are restricted.

\begin{theorem}\label{thm:triangle-star}
The multiple-unicast conjecture holds for every undirected network whose
demand graph is a subgraph of the union of a triangle and a star.
\end{theorem}

\begin{proof}
Theorem~\ref{thm:karzanov-packings} provides the required
\(K_2\)--\(K_{2,3}\) fractional metric packing. Corollary~
\ref{cor:k2-k23-packing} proves the claim.
\end{proof}

This is our third new network class. Unlike Theorems~\ref{thm:five}
and~\ref{thm:three-holes}, it places no bound on the number of terminal
locations and no planarity condition on the supply graph.

\subsection{A New Proof: At Most Six Coding Nodes}
\label{sec:six}

Let \(W\) be a fixed set of six labels. For
\(m\in\mathbb{R}_{\ge0}^{\binom{W}{2}}\), extend \(m\) by
\(m(x,x)=0\) and \(m(x,y)=m(y,x)\). Define the six-point metric cone by
\begin{equation}\label{eq:met6}
  \MET_6
  :=
  \left\{
    m\in\mathbb{R}_{\ge0}^{\binom{W}{2}} :
    \begin{array}{l}
      m(x,z)\le m(x,y)+m(y,z)\\[-1pt]
      \text{for all }x,y,z\in W
    \end{array}
  \right\}.
\end{equation}
Thus, \(\MET_6\) is the set of semimetrics on \(W\).

For a nonzero \(m\in\MET_6\), its ray is
\(\{\lambda m:\lambda\ge0\}\). This ray is \emph{extreme} if, whenever
\(m=m_1+m_2\) with \(m_1,m_2\in\MET_6\), there exists
\(\alpha\in[0,1]\) such that
\(m_1=\alpha m\) and \(m_2=(1-\alpha)m\).
\begin{table}[!h]
\caption{Graph-metric types in Grishukhin's six-point classification.}
\label{tab:six-metrics}
\centering
\footnotesize
\begin{tabular}{@{}lc@{}}
\toprule
Type & Maximum distance \\
\midrule
\(K_2\) & \(1\) \\
\(K_{2,3}\) with one pair identified & \(2\) \\
\(K_{2,4}\) or \(K_{3,3}\) & \(2\) \\
\(K_{3,3}-e\) & \(3\), attained by one pair \\
\bottomrule
\end{tabular}
\end{table}

For a connected unit-length graph \(Q\) and a map
\(\phi:W\to V(Q)\), write
\(m_{Q,\phi}(x,y):=d_Q\bigl(\phi(x),\phi(y)\bigr)\).
Grishukhin~\cite{grishukhin1992extreme} determined the extreme rays of
\(\MET_6\). After permuting the labels in \(W\), every extreme ray contains
a graph metric \(m_{Q,\phi}\), where
$
  Q\in\{K_2,K_{2,3},K_{2,4},K_{3,3},K_{3,3}-e\}.
$
For the \(K_{2,3}\) type, \(\phi\) identifies precisely one pair of labels:
it maps two labels to one vertex of \(K_{2,3}\) and maps the other four labels
to the remaining four vertices.

A standard extreme-ray decomposition for cones defined by finitely many linear
inequalities implies that every semimetric in \(\MET_6\) is a nonnegative
linear combination of semimetrics on its extreme rays. Table~\ref{tab:six-metrics}
lists the graph-metric types in the classification. Its second column gives
\(\max_{x,y\in W}m_{Q,\phi}(x,y)\) for the unit-length representative.

Fig.~\ref{fig:k33minus} displays \(K_{3,3}-e\), the only type in
Table~\ref{tab:six-metrics} with maximum distance greater than two. The
endpoints of its deleted edge have distance three, whereas every other
unordered pair has distance at most two.

\begin{figure}[!h]
\centering
\resizebox{.7\columnwidth}{!}{%
\begin{tikzpicture}[
  x=1.45cm,
  y=1.10cm,
  line cap=round,
  line join=round,
  vtx/.style={
    circle,
    draw=black,
    fill=white,
    line width=.65pt,
    minimum size=4mm,
    inner sep=0pt,
    font=\scriptsize
  },
  focusv/.style={
    circle,
    draw=red!60!black,
    fill=white,
    text=red!60!black,
    line width=.78pt,
    minimum size=4mm,
    inner sep=0pt,
    font=\scriptsize
  },
  edge/.style={
    draw=black,
    line width=.65pt
  },
  missing/.style={
    draw=red!60!black,
    densely dashed,
    line width=.78pt
  }
]

\coordinate (a1) at (0,1.40);
\coordinate (b1) at (1.50,.62);
\coordinate (a2) at (1.50,-.83);
\coordinate (b2) at (0,-1.52);
\coordinate (a3) at (-1.50,-.83);
\coordinate (b3) at (-1.50,.62);

\draw[edge] (a1)--(b1)--(a2)--(b2)--(a3);
\draw[edge] (b3)--(a1);
\draw[missing]
  (a3)--node[
    midway,
    left=5pt,
    align=right,
    font=\scriptsize,
    text=red!60!black
  ]{deleted edge\\distance \(3\)}(b3);

\draw[edge] (a1)
  .. controls (.34,.55) and (.27,-.72) .. (b2);
\draw[edge] (a2)
  .. controls (.82,-.48) and (-.88,.30) .. (b3);
\draw[edge] (a3)
  .. controls (-.73,-.18) and (.98,.43) .. (b1);

\node[vtx]    at (a1) {$a_1$};
\node[vtx]    at (a2) {$a_2$};
\node[focusv] at (a3) {$a_3$};
\node[vtx]    at (b1) {$b_1$};
\node[vtx]    at (b2) {$b_2$};
\node[focusv] at (b3) {$b_3$};

\end{tikzpicture}%
}
\caption{The graph \(K_{3,3}-e\) in
Table~\ref{tab:six-metrics}. The dashed red segment marks the deleted
edge; its two endpoints have distance three.}
\label{fig:k33minus}
\end{figure}
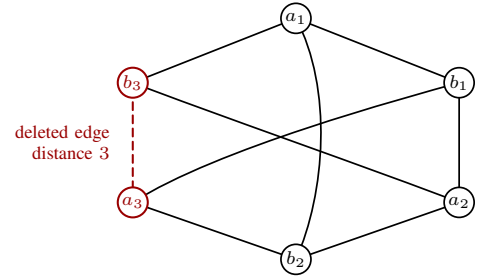

\begin{lemma}[Six-point graph-metric inequalities]
\label{lem:six-point-inequalities}
Every graph metric induced by a map to one of the five graphs listed in
Table~\ref{tab:six-metrics} satisfies \eqref{eq:metric-inequality}.
\end{lemma}

\begin{proof}
The graphs \(K_2\), \(K_{2,3}\), \(K_{2,4}\), and \(K_{3,3}\) are connected
bipartite graphs of diameter at most two. Therefore,
\eqref{eq:yin-unit} and the Transfer Lemma give the
graph-metric inequality.

In \(K_{3,3}-e\), only the endpoints of the deleted edge have distance greater
than two. Thus, every ordered endpoint pair at distance greater than two has
one of the two orientations of that unordered pair. Our bidirectional
long-pair extension, followed by the Transfer
Lemma, gives the graph-metric inequality in this case.
\end{proof}

\begin{theorem}\label{thm:six}
The multiple-unicast conjecture holds for every undirected network with at
most six coding nodes.
\end{theorem}

\begin{proof}
We use the forwarding-node reduction of
Yin~\etal~\cite[discussion preceding Corollary~4]{yin2018reduction}.
Because a forwarding node neither changes nor copies a symbol, each
transmission between coding nodes follows a walk whose internal nodes are
forwarding nodes. Replacing every such walk by a direct edge of its
shortest-path cost does not increase the coding cost and preserves all
terminal-pair distances. It therefore suffices to prove the assertion for an
arbitrary graph \(G\) with at most six nodes. Fix \(G\), a collection
\(\mathcal H\) of sessions, and a nonnegative edge-cost assignment \(\ell\).

If \(\lvert V(G)\rvert\le5\), then \(\mathcal H\) uses at most five terminal
locations, and Theorem~\ref{thm:five} shows that the conjecture holds for
\(G\). We may therefore assume that \(\lvert V(G)\rvert=6\). Identify \(W\)
with \(V(G)\), and let \(d_{G,\ell}\) be the shortest-path semimetric of
\((G,\ell)\).

The semimetric \(d_{G,\ell}\) belongs to \(\MET_6\). The extreme-ray
decomposition of \(\MET_6\), together with the six-point classification above,
gives graph metrics \(m_a\), each induced by a map to one of the five graphs
in Table~\ref{tab:six-metrics}, and coefficients \(\lambda_a\ge0\) such that
\begin{equation}\label{eq:six-point-decomposition}
  d_{G,\ell}(u,v)
  =
  \sum_a\lambda_a m_a(u,v)
  \qquad (u,v\in V(G)).
\end{equation}

For every supply edge \(uv\), shortest-path optimality gives
\(d_{G,\ell}(u,v)\le\ell(uv)\). Thus,
\eqref{eq:six-point-decomposition} gives a fractional metric packing for
\((G,\ell,\mathcal P(\mathcal H))\): it respects the cost of every supply
edge and gives equality for every terminal pair. By
Lemma~\ref{lem:six-point-inequalities}, every graph metric in this packing
satisfies \eqref{eq:metric-inequality}. Because \(\ell\) was arbitrary, the
hypotheses of Theorem~\ref{thm:divide-combine} are satisfied. Hence the
multiple-unicast conjecture holds for \(G\). The reduction then proves the
assertion of Theorem~\ref{thm:six}.
\end{proof}

This proof replaces the computer-aided search used in the earlier proof with
the finite classification of six-point metrics and the graph-metric
inequalities above.
\section{A Conditional \texorpdfstring{$\nu(D)\le2$}{nu(D)<=2} Result}
\label{sec:boundary}

We now consider instances whose demand graph has matching number at most
two. This condition excludes three sessions whose terminal locations are
pairwise disjoint, while allowing arbitrarily many terminal locations.

\subsection{Demand Graphs with
\texorpdfstring{$\nu(D)\le2$}{nu(D)<=2}}

After isolated vertices are discarded, every demand graph with matching number at most two is a subgraph
of one of four maximal patterns: a graph on at most five vertices, the union
of two stars, the union of \(K_3\) and a star, or \(K_3+K_3\)~\cite{hirai2010metric}. These cases are
illustrated in Fig.~\ref{fig:matching-two-patterns}. In the two star-union
cases, the constituent graphs need not be vertex-disjoint; \(K_3+K_3\)
denotes two vertex-disjoint triangles. Theorem~\ref{thm:karzanov-packings}
covers the first three cases. Thus, only the \(K_3+K_3\) case requires a
separate argument.

\begin{figure}[!htbp]
\centering
\resizebox{.9\columnwidth}{!}{%
\begin{tikzpicture}[
  line cap=round,
  line join=round,
  graphvertex/.style={
    circle,
    fill=black,
    minimum size=3.2pt,
    inner sep=0pt
  },
  graphedge/.style={
    draw=black,
    line width=.55pt
  },
  panel/.style={
    font=\footnotesize
  }
]

\begin{scope}[xshift=1.35cm,yshift=3.05cm]
  \node[panel] at (0,1.27) {(a) \(D\subseteq K_5\)};

  \foreach \i/\x/\y in {
    1/0/.68,
    2/.72/.16,
    3/.45/-.69,
    4/-.45/-.69,
    5/-.72/.16}
    {\coordinate (a\i) at (\x,\y);}

  \foreach \i/\j in {
    1/2,1/3,1/4,1/5,2/3,
    2/4,2/5,3/4,3/5,4/5}
    {\draw[graphedge] (a\i)--(a\j);}

  \foreach \i in {1,2,3,4,5}
    {\node[graphvertex] at (a\i) {};}
\end{scope}

\begin{scope}[xshift=5.15cm,yshift=3.05cm]
  \node[panel] at (0,1.27) {(b) \(D\subseteq S_1\cup S_2\)};

  \coordinate (bL) at (-.60,-.02);
  \coordinate (bR) at (.60,-.08);

  \coordinate (bC1) at (0,.30);
  \coordinate (bC2) at (.02,.10);
  \coordinate (bC3) at (0,-.30);

  \coordinate (bL1) at (-.43,.78);
  \coordinate (bL2) at (-.14,.50);
  \coordinate (bL3) at (-.12,-.68);
  \coordinate (bR1) at (.35,.74);
  \coordinate (bR2) at (.10,-.86);

  \draw[graphedge] (bL)--(bR);

  \foreach \n in {bC1,bC2,bC3}
    {
      \draw[graphedge] (bL)--(\n);
      \draw[graphedge] (bR)--(\n);
    }

  \foreach \n in {bL1,bL2,bL3}
    {\draw[graphedge] (bL)--(\n);}

  \foreach \n in {bR1,bR2}
    {\draw[graphedge] (bR)--(\n);}

  \foreach \n in {
    bL,bR,bC1,bC2,bC3,
    bL1,bL2,bL3,bR1,bR2}
    {\node[graphvertex] at (\n) {};}
\end{scope}

\begin{scope}[xshift=1.45cm,yshift=.05cm]
  \node[panel] at (.10,1.27) {(c) \(D\subseteq K_3\cup S\)};

  \coordinate (c1) at (-.78,.54);
  \coordinate (c2) at (-.78,-.54);
  \coordinate (c3) at (-.12,.00);
  \coordinate (cS) at (.43,.00);

  \coordinate (c4) at (.43,.67);
  \coordinate (c5) at (.98,.45);
  \coordinate (c6) at (1.15,.00);
  \coordinate (c7) at (.98,-.45);
  \coordinate (c8) at (.43,-.67);

  \draw[graphedge] (c1)--(c2);
  \draw[graphedge] (c2)--(c3);
  \draw[graphedge] (c3)--(c1);

  \foreach \n in {c1,c2,c3,c4,c5,c6,c7,c8}
    {\draw[graphedge] (cS)--(\n);}

  \foreach \n in {c1,c2,c3,cS,c4,c5,c6,c7,c8}
    {\node[graphvertex] at (\n) {};}
\end{scope}

\begin{scope}[xshift=5.15cm,yshift=.05cm]
  \node[panel] at (0,1.27) {(d) \(D\subseteq K_3+K_3\)};

  \coordinate (d1) at (-.72,.63);
  \coordinate (d2) at (-1.22,-.43);
  \coordinate (d3) at (-.22,-.43);

  \coordinate (d4) at (.72,.63);
  \coordinate (d5) at (.22,-.43);
  \coordinate (d6) at (1.22,-.43);

  \draw[graphedge] (d1)--(d2);
  \draw[graphedge] (d2)--(d3);
  \draw[graphedge] (d3)--(d1);

  \draw[graphedge] (d4)--(d5);
  \draw[graphedge] (d5)--(d6);
  \draw[graphedge] (d6)--(d4);

  \foreach \n in {d1,d2,d3,d4,d5,d6}
    {\node[graphvertex] at (\n) {};}
\end{scope}

\end{tikzpicture}%
}
\caption{Four graph families covering all demand graphs with matching
number at most two. After isolated vertices are deleted, every such graph
\(D\) is a subgraph of (a) \(K_5\), (b) \(S_1\cup S_2\),
(c) \(K_3\cup S\), or (d) \(K_3+K_3\), where \(S,S_1,S_2\) are
stars. The constituent graphs in (b) and (c) need not be vertex-disjoint;
only representative star leaves are shown.}
\label{fig:matching-two-patterns}
\end{figure}
\subsection{The \(K_3+K_3\) Case}

The known packing for \(K_3+K_3\) uses \(K_2\)-, \(K_{2,3}\)-,
\(K_{3,3}\)-, and \(\Gamma_{3,3}\)-metrics~\cite{hirai2010metric}. The
fixed bipartite graph \(\Gamma_{3,3}\), shown in
Fig.~\ref{fig:gamma33}, is obtained from \(K_{3,3}\) with bipartition
\(A=\{a_1,a_2,a_3\}\) and \(B=\{b_1,b_2,b_3\}\) by subdividing every edge
\(a_i b_j\) with a vertex \(x_{ij}\) and a vertex \(o\) adjacent to
all nine vertices \(x_{ij}\).

\begin{figure}[!bp]
\centering
\resizebox{.58\columnwidth}{!}{%
\begin{tikzpicture}[
  x=.72cm,y=.72cm,
  line cap=round,
  line join=round,
  original/.style={
    circle, draw=black, fill=black,
    minimum size=1.50mm, inner sep=0pt
  },
  subdivision/.style={
    circle, draw=black!72, fill=white,
    line width=.42pt, minimum size=1.50mm, inner sep=0pt
  },
  hub/.style={
    circle, draw=black, fill=white,
    line width=.90pt, minimum size=2.0mm, inner sep=0pt
  },
  outeredge/.style={draw=black, line width=.58pt},
  hubedge/.style={draw=black!28, line width=.38pt}
]

\coordinate (a1) at (0,1.55);
\coordinate (a2) at (4.05,.22);
\coordinate (a3) at (-.55,-2.02);

\coordinate (b1) at (2.25,1.55);
\coordinate (b2) at (-1.95,-.15);
\coordinate (b3) at (2.25,-2.02);

\coordinate (x11) at (1.13,1.57);
\coordinate (x12) at (-.95,.75);
\coordinate (x13) at (2.32,.08);

\coordinate (x21) at (3.15,.93);
\coordinate (x22) at (.42,-1.55);
\coordinate (x23) at (3.10,-.85);

\coordinate (x31) at (-.55,.07);
\coordinate (x32) at (-1.48,-1.12);
\coordinate (x33) at (1.15,-2.04);

\coordinate (o) at (1.15,-.38);

\foreach \x in {x11,x12,x13,x21,x22,x23,x31,x32,x33}
  {\draw[hubedge] (o)--(\x);}

\draw[outeredge] (a1)--(x11)--(b1);
\draw[outeredge] (a1)--(x12)--(b2);
\draw[outeredge] (a1)--(x13)--(b3);

\draw[outeredge] (a2)--(x21)--(b1);
\draw[outeredge] (a2)--(x22)--(b2);
\draw[outeredge] (a2)--(x23)--(b3);

\draw[outeredge] (a3)--(x31)--(b1);
\draw[outeredge] (a3)--(x32)--(b2);
\draw[outeredge] (a3)--(x33)--(b3);

\foreach \x in {x11,x12,x13,x21,x22,x23,x31,x32,x33}
  {\node[subdivision] at (\x) {};}

\node[hub] at (o) {};

\foreach \x in {a1,a2,a3,b1,b2,b3}
  {\node[original] at (\x) {};}

\end{tikzpicture}%
}
\caption{The graph \(\Gamma_{3,3}\). The six filled nodes are the vertices
of the original \(K_{3,3}\); each small outlined node subdivides one edge,
and the thick-outlined node is adjacent to all nine subdivision nodes.}
\label{fig:gamma33}
\end{figure}
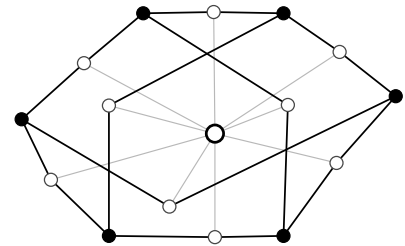

\begin{theorem}\label{thm:conditional-matching-two}
If the multiple-unicast conjecture holds for every multiple-unicast instance
on \(\Gamma_{3,3}\), then it holds for every multiple-unicast instance whose
demand graph has matching number at most two.
\end{theorem}

\begin{proof}
The first three patterns are covered by
Theorem~\ref{thm:karzanov-packings} and
Corollary~\ref{cor:k2-k23-packing}. It remains to consider
\(D(\mathcal H)\subseteq K_3+K_3\). Fix a nonnegative edge-cost assignment
\(\ell\) and use the \(K_3+K_3\) packing described above.

The hypothesis and Theorem~\ref{thm:cost-equivalence} give
\eqref{eq:yin-unit} with \(Q=\Gamma_{3,3}\). Every pair of vertices in
\(K_2\), \(K_{2,3}\), and \(K_{3,3}\) has distance at most two, so
\eqref{eq:yin-unit} applies on each of these graphs as well. Our Transfer
Lemma, applied to the map defining each metric in the
packing, therefore gives \eqref{eq:metric-inequality} for every such metric.
Since \(\ell\) was arbitrary, Theorem~\ref{thm:divide-combine} establishes
the conditional conclusion.
\end{proof}

Thus, within the \(\nu(D)\le2\) family, it remains to prove the
multiple-unicast conjecture for every collection of independent unicast
sessions on the fixed graph \(\Gamma_{3,3}\).

\section{Conclusion and Discussion}\label{sec:conclusion}

We develop a unified framework that identifies graph metrics, rather than cuts
alone, as the basic objects for comparing coding and routing. Its key
ingredient is our Transfer Lemma, which lifts coding--routing comparisons from
fixed graphs to arbitrary networks through the corresponding graph metrics.
When these metrics form a weighted decomposition of a network's shortest-path
metric, the lifted comparisons combine to prove that the conjecture holds
for that network.

Using this framework, we prove that the conjecture holds for networks with at
most five terminal locations, for three-hole planar networks in which the
endpoints of each session lie on the boundary of one hole, and for networks
whose demand graph is a subgraph of the union of a triangle and a star. We
also give a new proof that the conjecture holds for networks with at most six
coding nodes, replacing computer-aided search with the classification of
six-point metrics.

The framework also isolates two concrete open directions. Among demand graphs
with matching number at most two, the remaining maximal pattern is
\(K_3+K_3\), whose graph-metric decomposition includes
\(\Gamma_{3,3}\)-metrics~\cite{hirai2010metric}. Proving the conjecture on the
fixed graph \(\Gamma_{3,3}\) would settle all instances whose demand graph has
matching number at most two. The other direction concerns four-hole planar
networks, whose known decomposition uses a broader family of graph
metrics~\cite{karzanov1994metrics}. Establishing the corresponding
coding--routing comparisons would extend the framework to this planar setting.

\section*{Use of AI Disclosure}
OpenAI's ChatGPT was used to locate relevant literature and to turn
proof approaches conceived by the authors into initial written drafts.
It also assisted with implementing all figures in TikZ; the mathematical
content and visual design of the figures were determined by the authors.
The authors independently verified every mathematical step and source,
revised all AI-assisted text and code, and take full responsibility for
the manuscript.

\appendices

\section{Bidirectional Long-Pair Extension}
\label{app:bidirectional}

Lemma~\ref{lem:bidirectional} leaves one case to verify: for some nodes
\(a,b\in V(Q)\), every session whose source--receiver distance exceeds two
has ordered endpoints \((a,b)\) or \((b,a)\). This appendix gives the
required layer-counting argument.

Let \(Q=(V_Q,E_Q)\) be a connected bipartite graph with unit edge lengths,
and fix an arbitrary zero-error coding solution on \(Q\). For session \(i\),
let \(X_i\) be its independent source random variable. For every
\(uv\in E_Q\), let \(Y_{uv}\) denote the random variable transmitted from
\(u\) to \(v\). Measure entropy in the same unit as the rates, so
\(r_i=H(X_i)\) and \(f(uv)=H(Y_{uv})+H(Y_{vu})\).

For \(U\subseteq V_Q\), let \(U^{\mathrm{in}}\) and
\(U^{\mathrm{out}}\) be the collections of transmission random variables
entering and leaving \(U\). Write \(S_U:=\{X_i:s_i\in U\}\),
\(T_U:=\{X_i:t_i\in U\}\), and
\(\widehat U:=U^{\mathrm{in}}\cup S_U\). For a collection \(\mathcal A\) of
random variables, let \(\operatorname{Dom}(\mathcal A)\) be the set of
random variables determined by \(\mathcal A\). The Input--Output inequality
states that \(U^{\mathrm{out}}\cup T_U\subseteq
\operatorname{Dom}(\widehat U)\). The Crypto inequality states that
\(X_i\in\operatorname{Dom}(U^{\mathrm{in}}\cup U^{\mathrm{out}})\) whenever
\(\delta_Q(U)\) separates \(s_i\) and \(t_i\)~\cite{jain2006capacity}.
Whether the Crypto inequality applies depends only on the unordered terminal
pair \(\{s_i,t_i\}\), not on the orientation of the session.

Merge all sessions with the same ordered endpoints into one session. The
merged session transmits all of their independent source messages and has
rate equal to the sum of their rates, so~\eqref{eq:bidirectional-unit} is
unchanged. If at most one long ordered endpoint pair remains,
then~\eqref{eq:yin-unit} applies. Otherwise, the two long ordered endpoint
pairs are \((a,b)\) and \((b,a)\), with source variables \(X^+\) and
\(X^-\), respectively. Let \(Z:=(X^+,X^-)\), write
\(D:=d_Q(a,b)>2\), and let \(I_0\) index all remaining sessions. Thus every
session in \(I_0\) has source--receiver distance at most two, and
\(H(Z)=H(X^+)+H(X^-)\).

Let \(M:=\max_{v\in V_Q}d_Q(a,v)\), and set
\(U_j:=\{v\in V_Q:d_Q(a,v)=j\}\) for \(0\le j\le M\). Thus \(a\in U_0\) and
\(b\in U_D\). Since \(Q\) is bipartite and has unit edge lengths, every edge
joins consecutive layers. For \(1\le j\le M\), let
\(B_j:=\bigcup_{q<j}U_q\). Then \(\delta_Q(B_j)\) consists precisely of the
edges between \(U_{j-1}\) and \(U_j\).

Apply the layer-counting argument in the proof of~\eqref{eq:yin-unit} to
\(U_0,\ldots,U_M\), with \(X^+\) as the distinguished long source. Let
\(\mathcal L:=\sum_{uv\in E_Q}f(uv)+\sum_{i=1}^k H(X_i)\) denote the
left-hand side of the resulting combined entropy inequality. That argument
charges every directed transmission entropy at most once and includes one
initial copy of every source entropy.

Here, saying that \(H(X_i)\) is counted \(q\) times means that the
right-hand side contains \(q\) entropy terms that determine \(X_i\).
Because the source variables are independent, these terms contribute
\(qH(X_i)\) to the lower bound. The original layer count gives at least
\(d_Q(s_i,t_i)+1\) such right-hand-side terms for every \(i\in I_0\). We now
show that the \(D+1\) terms associated with \(U_0,\ldots,U_D\) determine
both \(X^+\) and \(X^-\).

For \(j=0\), \(X^+\in S_{U_0}\subseteq\widehat U_0\), while
\(X^-\in T_{U_0}\subseteq\operatorname{Dom}(\widehat U_0)\) by the
Input--Output inequality. Hence \(Z\in\operatorname{Dom}(\widehat U_0)\).
For \(j=D\), the roles of \(X^+\) and \(X^-\) are reversed, and similarly
\(Z\in\operatorname{Dom}(\widehat U_D)\).

Now let \(1\le j\le D-1\). The collection \(\widehat U_j\) contains the
transmissions entering \(U_j\) and determines the transmissions leaving
\(U_j\). It therefore determines both directed transmissions across
\(\delta_Q(B_j)\). Since this cut separates \(a\) and \(b\), the Crypto
inequality gives \(X^+,X^-\in\operatorname{Dom}(\widehat U_j)\), and hence
\(Z\in\operatorname{Dom}(\widehat U_j)\).

Write \(\mathcal A_j:=\operatorname{Dom}(\widehat U_j)\) for
\(0\le j\le D\). The \((D+1)\)-way submodularity step in the layer-counting
argument forms \(D+1\) right-hand-side entropy terms from unions of
intersections of \(\mathcal A_0,\ldots,\mathcal A_D\). Since \(Z\) belongs
to every \(\mathcal A_j\), every one of these \(D+1\) terms determines
\(Z\). Thus \(H(Z)\) is counted \(D+1\) times on the right-hand side.

Let \(\mathcal R_0:=\sum_{i\in I_0}
(d_Q(s_i,t_i)+1)H(X_i)\). The combined layer count gives
\begin{equation}\label{eq:app-final-count}
  \mathcal L
  \ge
  \mathcal R_0+(D+1)H(Z).
\end{equation}
Subtracting the initial copy of every source entropy from
\eqref{eq:app-final-count}, and using \(H(Z)=H(X^+)+H(X^-)\) and
\(D=d_Q(a,b)\), yields~\eqref{eq:bidirectional-unit}. Since the coding
solution was arbitrary, Lemma~\ref{lem:bidirectional} follows.
\begingroup
\renewcommand{\baselinestretch}{0.94}\selectfont
\bibliographystyle{IEEEtran}
\balance
\bibliography{references}
\endgroup
\end{document}